\documentclass[reqno]{amsart}

\usepackage{booktabs}
\usepackage{IEEEtrantools}
\usepackage{amssymb,latexsym,amsfonts,amsmath}
\usepackage{mathrsfs}
\usepackage{graphicx}
\usepackage{dsfont}
\usepackage{tikz}
\usetikzlibrary{calc,shapes,arrows}

\usepackage{algorithm}
\usepackage{algorithmic}

\usepackage{xspace}

\newtheorem{theorem}{Theorem}[section]
\newtheorem{lemma}[theorem]{Lemma}

\newtheorem{definition}[theorem]{Definition}

\newtheorem{remark}[theorem]{Remark}
\newtheorem{assumption}{Assumption}
\numberwithin{equation}{section}

\newcommand{\R}{{\mathbb{R}}}

\newcommand{\N}{{\mathbb{N}}}

\newcommand{\PP}{\mathds{P}}

\usepackage{fancyhdr}

\newenvironment{nouppercase}{%
	\renewcommand{\uppercasenonmath}[1]{}}{}

\begin{document}

\begin{abstract}
This article is concerned with a data-driven divide-and-conquer strategy to construct symbolic abstractions for interconnected control networks with unknown mathematical models. We employ a notion of \emph{alternating bisimulation functions} (ABF) to quantify the closeness between state trajectories of an interconnected network and its symbolic abstraction. Consequently, the constructed symbolic abstraction can be leveraged as a beneficial substitute for the formal verification and controller synthesis over the interconnected network. In our data-driven framework, we first establish a relation between each unknown subsystem and its data-driven symbolic abstraction, so-called \emph{alternating pseudo-bisimulation function (APBF)}, with a guaranteed probabilistic confidence. We then provide compositional conditions based on \emph{$\max$-type small-gain techniques} to construct an ABF for an unknown interconnected network using APBF of its individual subsystems, constructed from data. We demonstrate the efficacy of our data-driven approach over a room temperature network composing $100$ rooms with unknown models. We construct a symbolic abstraction from data for each room as an appropriate substitute of original system and compositionally synthesize controllers regulating the temperature of each room within a safe zone with some guaranteed probabilistic confidence.
\end{abstract}

\title{{\LARGE Symbolic Abstractions with Guarantees: A Data-Driven Divide-and-Conquer Strategy}}

\author{{\bf {\large Abolfazl Lavaei,~\emph{{\small Senior Member,~IEEE}}}}\\{\normalfont School of Computing, Newcastle University, United Kingdom}}

\pagestyle{fancy}
\lhead{}
\rhead{}
  \fancyhead[OL]{Abolfazl Lavaei}

  \fancyhead[EL]{Symbolic Abstractions with Guarantees: A Data-Driven Divide-and-Conquer Strategy}
  \rhead{\thepage}
 \cfoot{}
 
\begin{nouppercase}
	\maketitle
\end{nouppercase}

\section{Introduction}

Interconnected networks have been becoming popular during the past two decades as a valuable modeling scheme characterizing a broad range of real-world engineering systems. These networks find applications in  automated vehicles, drone networks, chemical networks, communication networks, and so on. In general, formal verification and controller design for this type of large-scale complex networks are computationally burdensome. This is especially due to (i) dealing with uncountable state/input sets with large dimensions, and (ii) absence of closed-form mathematical models in most of real-life scenarios. 

To alleviate these difficulties, one rewarding solution is to use symbolic abstractions as finite-state approximations of continuous-space models. By employing a constructed symbolic abstraction as an appropriate substitution of original (concrete) system, formal analyses can be performed over the abstract model. The acquired results can then be transfered back on the concrete domain, while quantifying a guaranteed error bound between state trajectories of two systems.  Accordingly, it can be guaranteed that the concrete system also satisfies the same specification as its symbolic abstraction within some quantified error bound~\cite{lavaei2022automated}.

There have been two variants of symbolic abstractions: \emph{sound} and \emph{complete}~\cite{tabuada2009verification}. Complete abstractions propose \emph{sufficient and necessary guarantees}: there exists a controller enforcing a desired property on a symbolic abstraction \emph{if and only if} there exists a controller satisfying the same specification over the original system. However, sound abstractions only provide \emph{sufficient guarantees}: not being able to synthesize a controller via a sound abstraction does not imply the lack of controller over the original domain.

There exist extensive results on abstraction-based analysis of control systems. Existing results encompass constructing (in)finite-abstractions for various classes of dynamical systems~\cite{tabuada2009verification,girard2009approximately,le2013mode,girard2015safety,coogan2015efficient}, to name a few.  However, constructing symbolic abstractions in a monolithic fashion  suffers significantly from  the \emph{curse of dimensionality} problem. To mitigate this computational complexity, \emph{compositional} abstraction-based techniques have received remarkable attentions to build a symbolic abstraction for an interconnected network using those of smaller subsystems~\cite{tazaki2008bisimilar,pola2016symbolic,swikir2019compositional,zamani2017compositional}. 

The above-mentioned studies on the construction of symbolic abstractions unfortunately require knowing precise dynamics of underlying systems. Although \emph{indirect data-driven} approaches strive to learn unknown dynamics via identification techniques~\cite{Hou2013model}, obtaining an accurate mathematical model is generally computationally challenging especially if the unknown system is complex. In addition, even if a model can be identified via system identification approaches, the relation between the identified model and its symbolic abstraction should be still constructed. Accordingly, the underlying complexity exists in two levels of model identification and establishing the relation. In this work, we develop a \emph{direct data-driven} scheme, without performing any model identification, and construct symbolic abstractions together with their associated similarity relations by directly gathering data from trajectories of unknown concrete systems.

The original contribution of this work is to develop a data-driven divide-and-conquer strategy for constructing symbolic abstractions for unknown interconnected networks while providing a guaranteed probabilistic confidence. The proposed approach relies on a notion of \emph{alternating bisimulation functions} (ABF) to quantify the closeness between trajectories of an interconnected network and its symbolic abstraction. In our data-driven scheme, we first recast conditions of \emph{alternating pseudo-bisimulation functions} (APBF) as a robust optimization program (ROP). By gathering samples from trajectories of each unknown subsystem, we provide a scenario optimization program (SOP) for each original ROP. We construct APBF from data with a guaranteed probabilistic confidence by establishing a probabilistic bridge between optimal values of SOP and ROP. We then propose a compositional approach using $\max$-type small-gain reasoning to construct an ABF for an unknown interconnected network via data-driven APBF of smaller subsystems. In fact, our data-driven divide and conquer approach resolves the sample complexity problem existing in almost all data-driven approaches whose main goal is to certify
some properties over unknown systems via data. In particular, the number of
data for providing formal analysis over unknown systems is \emph{exponential}
with respect to the size of the underlying system. However, the sample complexity in our compositional approach is reduced to subsystems: the
number of samples \emph{linearly} increases with the number of individual subsystems. We verify our data-driven results over a room temperature network composing $100$ rooms with unknown models.

There has been a limited number of work on the construction of symbolic abstractions using data. Existing results include: construction of symbolic abstractions via a Gaussian process approach~\cite{hashimoto2020learning}; data-driven abstraction of monotone systems with disturbances~\cite{makdesi2021efficient}; data-driven growth bound computation for constructing finite abstractions~\cite{kazemi2022data}; data-driven construction of symbolic abstractions for verification of unknown systems~\cite{coppola2022data}; and data-driven construction of finite abstractions for incrementally input-to-state stable systems~\cite{Lavaei_LCSS22_1}. In comparison, we propose a \emph{compositional data-driven} framework using small-gain reasoning for constructing symbolic abstractions of \emph{large-scale interconnected networks}, whereas the results
in~\cite{hashimoto2020learning,makdesi2021efficient,kazemi2022data,coppola2022data,Lavaei_LCSS22_1} are all tailored to \emph{monolithic systems}. As a result, the proposed approaches in~\cite{hashimoto2020learning,makdesi2021efficient,kazemi2022data,coppola2022data,Lavaei_LCSS22_1} suffer from the sample complexity problem and are not useful in practice when dealing with high-dimensional systems. In addition, the works~\cite{hashimoto2020learning,makdesi2021efficient,kazemi2022data,coppola2022data} construct \emph{sound} abstractions based on data (\emph{sufficient guarantees}), whereas our data-driven technique is for the construction of \emph{complete} abstractions (\emph{sufficient and necessary guarantees}).

\section{Discrete-Time Nonlinear Control Systems}\label{Sec: dt-NDS}

\subsection{Notation}

In this work, $\mathbb{R},\mathbb{R}^+$, and $\mathbb{R}^+_0$, represent sets of real, positive, and non-negative real numbers, respectively. Symbols $\mathbb{N} := \{0,1,2,...\}$ and $\mathbb{N}^+=\{1,2,...\}$ denote, respectively, sets of non-negative and positive integers. A column vector, given $N$ vectors $x_i \in \mathbb{R}^{n_i}$, is represented by $x=[x_1;\dots;x_N]$. We denote the minimum and maximum eigenvalues of a symmetric matrix $P$, respectively, by $\lambda_{\min}(P)$ and $\lambda_{\max}(P)$. Given two sets $X$ and $Y$, $\mathscr{R}\subseteq X \times Y$ denotes a relation which relates $x \in X$ to $y \in Y$ if $(x, y) \in \mathscr{R}$, equivalently $x\mathscr{R}y$. Given any scalar $a\in\mathbb R$ and vector $x\in\mathbb{R}^{n}$, $\vert a\vert$ and $\Vert x\Vert$ represent, respectively, the absolute value and the infinity norm. For a matrix $P\in\mathbb R^{m\times n}$, $\|P\| := \sqrt{\lambda_{\max}(P^\top P)}$. Given a probability space $(\mathcal D,\mathbb B(\mathcal D),\PP)$, we denote by $\mathcal D^N$ the $N$-Cartesian product of set $\mathcal D$, and by $\PP^N$ its corresponding product measure. A Gamma function $\Gamma$ is defined as $\Gamma(a) = (a - 1)!$ for any positive integer $a$ and $\Gamma(a + \frac{1}{2}) = (a - \frac{1}{2})\times (a - \frac{3}{2})\times \dots \times \frac{1}{2}\times \pi^{\frac{1}{2}}$ for any non-negative integer $a$. We show the feasibility of an optimization problem by $\vDash$.

\subsection{Discrete-Time Nonlinear Control Systems}
We first present the formal definition of discrete-time nonlinear control systems. 

\begin{definition}\label{dt-NCS}
	A discrete-time nonlinear control system (dt-NCS) is characterized by
	\begin{align}\label{EQ:1}
	\Xi=(X,U,D,f),
	\end{align}
	where:
	\begin{itemize}
		\item $X\subseteq \mathbb R^n$ is a state set;
		\item $U = \{\nu_1,\nu_2,\dots,\nu_m\}$ with $\nu_i \in \mathbb R^{\bar m}, i \in\{1,\dots,m\}$, is a finite input set;
		\item  $D\subseteq \mathbb R^p$ is a disturbance set;
		\item $f:X\times U\times D\rightarrow X$ is a transition map, which is \emph{unknown} in our setting.
	\end{itemize}
\end{definition}
The evolution of dt-NCS can be described by
\begin{align}\label{EQ:2}
\Xi\!:x(k+1)=f(x(k),\nu(k),d(k)),\quad k\in\mathbb N,
\end{align} 
for any $x\in X$, $\nu(\cdot):\mathbb N\rightarrow U$, and $d(\cdot):\mathbb N\rightarrow D$. The \emph{state trajectory} of $\Xi$ under sequences $\nu(\cdot), d(\cdot)$ starting from $x(0)= x_0$ is denoted by $x_{x_0\nu d}:\mathbb N\rightarrow X$.

Since the ultimate goal is to construct a symbolic abstraction for a network of dt-NCS, we consider the system in~\eqref{EQ:1} as a \emph{subsystem} and provide another definition for the \emph{interconnected} dt-NCS without disturbances $d$ which is acquired as a composition of individual subsystems with disturbances $d$.

\begin{definition}\label{Def:2}
	Consider $\mathcal M\!\in\!\mathbb N^+$ dt-NCS $\Xi_i\!=\!(X_i,U_i,D_i,f_i)$, $i\in \{1,\dots,\mathcal M\}$, with their disturbances partitioned as	
	\begin{align}\label{Eq:3}
	d_i&=[{d_{i1};\ldots;d_{i(i-1)};d_{i(i+1)};\ldots;d_{i\mathcal M}}].
	\end{align}
	An interconnected dt-NCS is defined as
	$\Xi=(X,U,f)$, represented by
	$\mathcal{I}(\Xi_1,\ldots, \Xi_{\mathcal M})$, where $X:=\prod_{i=1}^{\mathcal M}X_i$, $U:=\prod_{i=1}^{\mathcal M}U_i$, and $f:=[f_1;\dots;f_{\mathcal M}]$, such that:
	\begin{align}\label{Eq:4}
	\forall i,j\in \{1,\dots,\mathcal M\},i\neq j\!: \quad d_{ij}=x_{j}, ~~ X_{j}\subseteq D_{ij},
	\end{align}
	where $D_i:=\prod_{j\neq i}D_{ij}$. Such an interconnected dt-NCS is characterized by
	\begin{equation}\label{Eq:5}
	\Xi\!:x(k+1)=f(x(k),\nu(k)), \quad \text{where}~ f: X \times U \rightarrow X.
	\end{equation}
\end{definition}

\subsection{Symbolic Abstractions}\label{Symbolic}

Here, we construct symbolic abstractions as finite-state approximations of dt-NCS~\cite{pola2016symbolic}. To do so, state and disturbance sets are assumed to be compact. For constructing symbolic abstractions, we first partition state and disturbance sets as $X = \cup_i \mathsf X_i$ and $D = \cup_i \mathsf D_i$, and then pick representative points $\hat x_i\in \mathsf X_i$ and $\hat d_i\in \mathsf D_i$ within those partition sets as finite states and disturbances. In the next definition, we formally present how to construct symbolic abstractions.

\begin{definition}\label{Symbolic_Abs}
	Consider a dt-NCS $\Xi=(X,U,D,f)$ in~\eqref{EQ:1}. The constructed \emph{symbolic abstraction} $\hat\Xi$ is characterized as
	\begin{equation*}
	\hat\Xi =(\hat X, U, \hat D,\hat f),
	\end{equation*}
	where $\hat X$ and $\hat D$ are discrete state and disturbance sets of~$\hat\Xi$. Furthermore, $\hat f:\hat X \times U\times \hat D\rightarrow\hat X$ is a transition function defined as 
	\begin{equation}\label{EQ:3}
	\hat f(\hat{x}, \nu, \hat d) = \mathcal P(f(\hat{x}, \nu, \hat d)),
	\end{equation}
	where $\mathcal P\!\!:X\rightarrow \hat X$ is a quantization map with \emph{state discretization parameter} $\sigma$ fulfilling the following condition:
	\begin{equation}\label{EQ:4}
	\Vert \mathcal P(x)-x\Vert \leq \sigma,\quad \forall x\in X. 
	\end{equation}
\end{definition}

\section{Alternating (Pseudo-)Bisimulation Functions}\label{Sec:ABF}
In this section, we define notions of alternating pseudo-bisimulation and bisimulation functions for, respectively, dt-NCS and its symbolic abstraction (with disturbance signals) and two interconnected dt-NCS (without disturbance signals)~\cite{swikir2019compositional}.

\begin{definition}\label{Def:11}
	Consider a dt-NCS $\Xi =(X,U,D,f)$ as in Definition~\ref{dt-NCS} and its 
	symbolic abstraction $\hat\Xi =(\hat X,U,\hat D,\hat f)$ as in Definition~\ref{Symbolic_Abs}. A function $\mathcal S:X\times\hat X\to\R_0^+$ is an alternating pseudo-bisimulation function (APBF) between $\hat\Xi$ and $\Xi$, represented by $\hat\Xi\cong_{\mathcal{S}}\Xi$, if
	\begin{subequations}
		\begin{align}\label{EQ:51}
		&\forall x\in X, \forall \hat x\in\hat X\!: \quad\quad\quad\quad\quad\quad\quad\quad\quad\quad\quad\quad\quad\quad\quad\quad\quad\quad\gamma\Vert x - \hat x\Vert^2\leq \mathcal S(x,\hat x),\\\label{EQ:61}
		&\forall x\in X, \forall \hat x\in\hat X\!,  \forall\nu\in U\!, \forall d\in D, \forall\hat d\in\hat D\!: \quad\quad\mathcal S(f(x,\nu,d),\hat{f}(\hat x,\nu, \hat d))\leq\max\big\{\mu\mathcal S(x,\hat{x}), \eta\Vert d - \hat d\Vert^2,\theta\big\},
		\end{align}
	\end{subequations}
	for some $\gamma\in\R^+$, $0< \mu <1,$ and $\eta,\theta\in\R_0^+$.
\end{definition}
We now amend the above notion and present it as a relation between two interconnected dt-NCS by eliminating disturbance signals.

\begin{definition}\label{Def:1}
	Consider an interconnected dt-NCS $\Xi =(X,U,f)$ and its 
	symbolic abstraction $\hat\Xi =(\hat X,U,\hat f)$. A function $\mathcal V:X\times\hat X\to\R_0^+$ is an alternating bisimulation function (ABF) between $\hat\Xi$ and $\Xi$, denote by $\hat\Xi\cong_{\mathcal{V}}\Xi$, if
	\begin{subequations}
		\begin{align}\label{EQ:5}
		&\forall x\in X, \forall \hat x\in\hat X\!:\quad\quad\quad\quad\quad\quad\quad\quad\quad\quad\quad\quad\gamma\Vert x - \hat x\Vert^2\leq \mathcal V(x,\hat x),\\\label{EQ:6}
		&\forall x\in X, \forall \hat x\in\hat X\!, \forall\nu\in U\!:
		 \quad\quad\quad\quad\quad\mathcal V(f(x,\nu),\hat{f}(\hat x,\nu))\leq\max\big\{\mu\mathcal V(x,\hat{x}),\theta\big\},
		\end{align}
	\end{subequations}
	for some $\gamma\in\R^+$, $0< \mu <1,$ and $\theta \in\R_0^+$.
\end{definition}

The alternating bisimulation function in Definition~\ref{Def:1} implies that if the original dt-NCS and its symbolic abstraction commence from two close states (ensured by~\eqref{EQ:5}), then they stay close after a one-step evolution (ensured by~\eqref{EQ:6})~\cite{tabuada2009verification}.

In the next theorem, we leverage the usefulness of ABF and capture the distance between trajectories of an interconnected dt-NCS and its symbolic abstraction~\cite{swikir2019compositional}. 

\begin{theorem}\label{Thm:1}
	Given an interconnected dt-NCS $\Xi$ and its 
	symbolic abstraction $\hat\Xi$, let $\mathcal V$ be an ABF between $\hat\Xi$ and $\Xi$. Then a relation $\mathscr{R} \subseteq X \times \hat X$ as
	\begin{align}\label{Relation}
	\mathscr{R} := \Big\{(x,\hat x) \in X \times \hat X \,\big|\, \mathcal V(x,\hat{x}) \leq\theta\Big\}
	\end{align}
	is an $\tilde\epsilon$-approximate alternating bisimulation relation~\cite{tabuada2009verification} between $\hat\Xi$ and $\Xi$ with $\tilde\epsilon = (\frac{\theta}{\gamma})^{\frac{1}{2}}$.
\end{theorem}

In the next sections, we first construct APBF from data between unknown subsystems and their symbolic abstractions. We then provide sufficient compositional conditions in Section~\ref{Sec:Com} using a small-gain approach to construct an ABF for an interconnected system via its data-driven APBF of subsystems. 

\section{Data-Driven APBF}\label{DD-ABF}

In our data-driven approach, we consider APBF as $\mathcal{S}(\varphi,x,\hat x)=\sum_{j=1}^{z} {\varphi}_jg_j(x,\hat x)$, where $g_j(x,\hat x)$ are basis functions and $\varphi=[{\varphi}_{1};\dots;\varphi_z] \in \mathbb{R}^z$ are unknown variables. By considering basis functions $g_j(x,\hat x)$ as monomials over $(x,\hat x)$, APBF will be polynomial-type. To enforce proposed conditions of APBF as~\eqref{EQ:51}-\eqref{EQ:61}, we cast them as the following robust optimization program (ROP):
\begin{align}\label{ROP}
&\text{ROP}\!:\!\left\{
\hspace{-1.5mm}\begin{array}{l}\min\limits_{[\mathcal G;\xi]} \quad\!\xi,\\
\, \text{s.t.} \quad  \,\max_j\big\{\mathcal H_j(x, \hat x,\nu,d,\hat d, \mathcal G)\big\}\leq \xi,  j\in\{1,2\}, \\ 
\quad\quad\quad\!\forall x\in X,\forall \hat x\in \hat X,\forall \nu\in U,\forall d\in D, \forall \hat d\in \hat D,\\
\quad\quad\quad\!\mathcal G = [\gamma;\tilde\mu;\tilde\eta;\tilde\theta;{\varphi}_{1};\dots;\varphi_z],\\
\quad\quad\quad\!\gamma\!\in\!\R^+, \tilde\mu \in (0,1), \tilde\eta,\tilde\theta\!\in\!\R_0^+, \xi\!\in\!\mathbb R,\end{array}\right.
\end{align}
where: 
\begin{align}\notag
\mathcal H_1& = \gamma\Vert x - \hat x\Vert^2 - \mathcal S(\varphi,x,\hat x),\\\label{EQ:11}
\mathcal H_2& = \mathcal S(\varphi,f(x,\nu,d),\hat{f}(\hat x,\nu,\hat d))- \tilde\mu\mathcal S(\varphi,x,\hat{x})- \tilde\eta\Vert d - \hat d\Vert^2-\tilde\theta.
\end{align}
One can readily verify that conditions~\eqref{EQ:51}-\eqref{EQ:61} in the construction of APBF are fulfilled if $\xi_{\mathcal R}^* \leq 0$, with $\xi_{\mathcal R}^*$ being an optimal value for ROP.

\begin{remark}
	Note that after solving ROP in~\eqref{ROP}, $\mu,\eta,\theta$ in the $\max$-form condition~\eqref{EQ:61} can be acquired based on $\tilde\mu,\tilde\eta,\tilde\theta$ in the implication-form in~\eqref{EQ:11} as $\mu = 1 - (1 - \psi)(1 - \tilde\mu),\eta = \frac{(1+\lambda)\tilde \eta}{(1 - \tilde\mu)\psi}, \theta = \frac{(1+\frac{1}{\lambda})\tilde \theta}{(1 - \tilde\mu)\psi}$, for any $0<\psi <1$ and $\lambda \in \mathbb{R}^+$.
\end{remark}\vspace{0.2cm}
The provided ROP in~\eqref{ROP} is not solvable due to appearing unknown maps $f,\hat f$ in $\mathcal H_2$. To resolve this issue, we collect $\mathcal Q$ independent-and-identically distributed (i.i.d.) samples within $X\times D$, denoted by $(\bar x_i,\bar d_i)^{\mathcal Q}_{i=1}$. Now we propose a scenario optimization program (SOP), with an optimal value $\xi_{\mathcal Q}^*$, associated to the original ROP:
\begin{align}\label{SOP}
&\text{SOP}\!:\!\left\{
\hspace{-1.5mm}\begin{array}{l}\min\limits_{[\mathcal G;\xi]} \,\,\,\,\,\xi,\\
\, \text{s.t.}  \quad \,\max_j\!\big\{\mathcal H_j(\bar x_i,\hat x,\nu,\bar d_i,\hat d,\mathcal G)\big\}\leq \xi, j\in\{1,2\}, \\
\quad \quad\quad \!\!\forall \bar x_i\in X,\forall \bar d_i\in D,\forall i\in \{1,\ldots,\mathcal Q\},\\
\quad\quad\quad\!\!\forall \hat x\in \hat X,\forall \hat d\in \hat D,\forall \nu\in U,\\
\quad\quad\quad\!\!\mathcal G = [\gamma;\tilde\mu;\tilde\eta;\tilde\theta;{\varphi}_{1};\dots;\varphi_z],\\
\quad\quad\quad\!\gamma\!\in\!\R^+, \tilde\mu \in (0,1), \tilde\eta,\tilde\theta\!\in\!\R_0^+,\xi\!\in\!\mathbb R.\end{array}\right.
\end{align}
One can now substitute unknown $f(\bar x_i,\nu, \bar d_i)$ in $\mathcal H_2$ by measuring one-step transition of dt-NCS starting from $\bar x_i$ under $\nu$ and $\bar d_i$. As for $\hat f(\hat x,\nu,\hat d)$ in $\mathcal H_2$, we first compute $f(\hat x,\nu,\hat d)$ by initializing the unknown model from $\hat x$ under $\nu$ and $\hat d$. Given a discretization parameter $\sigma$, we then compute $\hat f(\hat x,\nu,\hat d)$ as the \emph{nearest point} to $f(\hat x,\nu,\hat d)$ by fulfilling condition~\eqref{EQ:4}. This is the way that we \emph{construct data-driven symbolic abstractions} by including a discretization error that is captured via $\theta$ in~\eqref{EQ:61}.

\begin{remark}
	Given a bilinearity between unknown variables $\varphi$ and $\tilde\mu$ in condition $\mathcal H_{2}$, we consider $\tilde\mu$ in a discrete set as $\tilde\mu \in \{\tilde\mu_1,\dots,\tilde\mu_l\}$. The cardinality $l$ is then incorporated in computing the required number of data for solving SOP (cf.~\eqref{EQ:12}). 
\end{remark}

\section{Data-Driven Guarantee for APBF Construction}

In this section, via the next theorem, we construct an APBF between each unknown subsystem and its symbolic abstraction with a guaranteed probabilistic confidence by establishing a probabilistic bridge between optimal values of SOP and ROP~\cite{esfahani2014performance}.

\begin{theorem}\label{Thm:3}
	Given an unknown dt-NCS in~\eqref{EQ:2}, let $\mathcal H_1$ and $\mathcal H_2$ be Lipschitz continuous with respect to $x$ and $(x,d)$ with, respectively, Lipschitz constants $\mathscr{L}_{1}$, $\mathscr{L}_{{2_t}}$, for given $\tilde\mu_t$ where $ t\in\{1,\dots,l\}$, and any $\nu\in U$. Consider the $\text{SOP}$ in~\eqref{SOP} with $\xi^*_{\mathcal Q}$, $\mathcal G^* = [\gamma^*;\tilde\eta^*;\tilde\theta^*,{\varphi}^*_{1};\dots;\varphi^*_z]$, and	
	\begin{align}\label{EQ:12}
	\mathcal Q(\varepsilon_t,\beta):=\min\Big\{\mathcal Q\in\N \,\big|\sum_{t=1}^{l}\sum_{i=0}^{c-1}\binom{\mathcal Q}{i}\varepsilon^i_t(1-\varepsilon_t)^{\mathcal Q-i}\leq\beta\Big\},
	\end{align}
	where $\beta,\varepsilon_t\in [0,1]$ for any $t\in\{1,\dots,l\}$, with $c,l$ being, respectively, number of unknown variables in $\text{SOP}$, and cardinality of finite set of $\tilde\mu$. If 
	\begin{align}\label{Con1}
	\xi^*_{\mathcal Q}+\max_t\mathscr{L}_{\mathcal H_t}\varkappa^{-1}(\varepsilon_t) \leq 0,
	\end{align}	
	with $\mathscr{L}_{\mathcal H_t} = \max\{\mathscr{L}_{1}, \mathscr{L}_{{2_t}}\}$, and $\varkappa(s): \mathbb R^+_0\rightarrow [0,1]$ depending on the geometry of $X\times D$ and the sampling distribution, then the constructed $\mathcal S$ via data is an APBF between $\hat\Xi$ and $\Xi$ with a guaranteed confidence of $1-\beta$, \emph{i.e.,} $\PP^{\mathcal Q}\big\{\hat\Xi\cong_{\mathcal{S}}\Xi\big\}\ge 1-\beta.$
\end{theorem}

\begin{proof}
	According to~\cite[Theorem 4.3]{esfahani2014performance}, one can quantify the closeness between optimal values of
	ROP and SOP as
	\begin{align}\label{EQ:12_1}
	\PP^{\mathcal Q} \Big\{0\leq\xi^*_{\mathcal R}-\xi^*_{\mathcal Q}\leq\max_t\bar\varepsilon_t\Big\}\geq 1-\beta,
	\end{align}
	with $$\mathcal {Q}\big(\varkappa(\frac{\bar\varepsilon_t}{\mathrm L_{\mathrm {SP}}\mathscr{L}_{\mathcal H_t}}),\beta\big),$$ where $\bar\varepsilon_t\in [0,1]$, $\varkappa(s): \mathbb R^+_0\rightarrow [0,1]$, and $\mathrm{L}_{\mathrm{SP}}$ is a Slater point which is considered here as $1$ given that the original ROP in~\eqref{ROP} is a $\min$-$\max$ optimization program~\cite[Remark 3.5]{esfahani2014performance}.
	
	\noindent
	From~\eqref{EQ:12_1}, one has $\xi^*_{\mathcal Q}\leq\xi^*_{\mathcal R}\leq\xi^*_{\mathcal Q}+\max_t\bar\varepsilon_t$ with a confidence of $1-\beta$. If $\xi^*_{\mathcal Q}+\max_t\bar\varepsilon_t \leq 0$, then $\xi^*_{\mathcal R} \leq 0$, implying that  conditions~\eqref{EQ:51}-\eqref{EQ:61} are satisfied and the constructed $\mathcal S$ from data is an APBF between $\hat\Xi$ and $\Xi$ with a confidence of at least $1-\beta$. Since $\varepsilon_t=\varkappa(\frac{\bar\varepsilon_t}{\mathscr{L}_{\mathcal H_t}})$ with $\mathrm{L}_{\mathrm{SP}} = 1$~\cite{esfahani2014performance}, one has $\bar\varepsilon_t = \mathscr{L}_{\mathcal H_t}\varkappa^{-1}(\varepsilon_t)$. Then one can recast condition $\xi^*_{\mathcal Q}+\max_t\bar\varepsilon_t \leq 0$ as $\xi^*_{\mathcal Q}+\max_t\mathscr{L}_{\mathcal H_t}\varkappa^{-1}(\varepsilon_t) \leq 0$, which completes the proof. \end{proof}

In the next lemma, we compute the function $\varkappa$ in~\eqref{Con1} when collecting data with a uniform sampling distribution from a hyper-rectangle uncertainty set. 

\begin{figure*}
	\begin{align}\notag
	&\mathscr{L}_{{2_t}} = \max\limits_{x\in X, d\in D}\Vert \begin{bmatrix}
	2((Ax + B\nu+ Ed) - \mathcal P(A\hat x+ B\nu+ E\hat d))^\top PA- 2\tilde\mu_t(x-\hat x)^\top P\\2((Ax + B\nu+ Ed) - \mathcal P(A\hat x+ B\nu+ E\hat d))^\top PE-2\eta(d-\hat d)^\top
	\end{bmatrix}\Vert\\\notag
	&\leq \max\limits_{x\in X, d\in D}\Big\{\Vert  2((Ax + B\nu+ Ed) - \mathcal P(A\hat x+ B\nu+ E\hat d))^\top PA\Vert +\Vert 2\tilde\mu_t(x-\hat x)^\top P\Vert\\\notag
	&~~~+ \Vert  2((Ax + B\nu+ Ed) - \mathcal P(A\hat x+ B\nu+ E\hat d))^\top PE\Vert +\Vert 2\eta(d-\hat d)^\top\Vert\Big\}\\\notag
	& \leq  \max\limits_{x\in X, d\in D} \Big\{2\Vert P \Vert\big(\Vert A \Vert(\Vert A x\Vert + \Vert B  \nu \Vert + \Vert Ed \Vert + \Vert \mathcal P(A\hat x + B \nu + E\hat d)\Vert)\\\notag
	&~~~+ \Vert D \Vert(\Vert A x\Vert + \Vert B  \nu \Vert + \Vert Ed \Vert + \Vert \mathcal P(A\hat x + B \nu + E\hat d)\Vert)+\tilde\mu_t(\Vert x \Vert + \Vert \hat x \Vert)\big)+2\eta(\Vert d\Vert + \Vert\hat d\Vert)\Big\}\\\notag
	&\leq\max\limits_{x\in X, d\in D} \Big\{2\Vert P \Vert\big(\Vert A \Vert(\Vert A x\Vert + \Vert B  \nu \Vert + \Vert Ed \Vert + \sigma + \Vert A\hat x + B \nu + E\hat d\Vert)\\\notag
	&~~~+ \Vert D \Vert(\Vert A x\Vert + \Vert B  \nu \Vert + \Vert Ed \Vert + \sigma + \Vert A\hat x + B \nu + E\hat d\Vert)+\tilde\mu_t(\Vert x \Vert + \Vert \hat x \Vert)\big)+ 2\eta(\Vert d\Vert + \Vert\hat d\Vert)\Big\}\\\label{EQ:100}
	& \leq 2\lambda_{\max}(P) (2\mathcal J_1^2 \varpi_1 \!+\! 2\mathcal J_1 \mathcal J_2 \varpi_2 \!+\! 2\mathcal J_1 \mathcal J_3 \varpi_3 \!+\! \mathcal J_1\sigma \!+\! 2\mathcal J_3^2 \varpi_3 \!+\! 2\mathcal J_2 \mathcal J_3 \varpi_2 \!+\! 2\mathcal J_1 \mathcal J_3 \varpi_1 \!+\! \mathcal J_3\sigma \!+\! 2\varpi_1\tilde\mu_t) \!+\! 2\tilde\eta\varpi_3.
	\end{align}
	\rule{\textwidth}{0.1pt}
\end{figure*}

\begin{lemma}\label{Function_g}
	The function $\varkappa$ in~\eqref{Con1} fulfills the following inequality~\cite[Proposition 3.8]{esfahani2014performance}: 
	\begin{align}\label{New}
	\varkappa(r) \leq \PP \big [\mathbb B_r(x,d)\big],\quad\quad \forall r\in\mathbb R^+_0, \forall (x,d) \in X\times D,
	\end{align}
	with $\mathbb B_r(a) \subset X\times D$ being an open ball with center $a$ and radius $r$. If one collects data from an $(n+p)$-dimensional \emph{hyper-rectangle} uncertainty set $X\times D$ with a \emph{uniform} distribution, then $\varkappa$ in~\eqref{New} is quantified as 
	\begin{align}\notag
	\varkappa(r) &=\frac{\text{Vol}(\mathbb B_r(x,d))}{2^{n+p}\text{Vol}(X\times D)} = \frac{\frac{\pi^{\frac{{n+p}}{2}}}{\Gamma(\frac{{n+p}}{2} + 1)}r^{n+p}}{2^{n+p}\text{Vol}(X\times D)}\\\label{g-function}
	& = \frac{\pi^{\frac{{n+p}}{2}}r^{n+p}}{2^{n+p}\Gamma(\frac{{n+p}}{2} + 1)\text{Vol}(X\times D)}, 
	\end{align}
	with $\text{Vol}(\cdot)$ and $\Gamma$ being volume set and Gamma function, respectively. For other types of sample distributions and uncertainty sets, the function $\varkappa$ can be computed according to~\cite{kanamori2012worst}.
\end{lemma}

To check the proposed condition in~\eqref{Con1}, $\mathscr{L}_{\mathcal H_t}$ is required. In the next lemmas, we compute $\mathscr{L}_{\mathcal H_t}$ for both linear and nonlinear control systems.

\begin{lemma}\label{Lem:1}
	Given a linear system $x(k+1)=Ax(k) + B\nu(k) + Ed(k)$, let $(x-\hat x)^\top P(x-\hat x)$ be an APBF with a positive-definite matrix $P\in\mathbb{R}^{n\times n}$. Then $\mathscr{L}_{\mathcal H_t}$ is computed as $\mathscr{L}_{\mathcal H_t} = \max\! \big\{\mathscr{L}_{1}, \mathscr{L}_{{2_t}}\big\}$ with 
	\begin{align*}
	\mathscr{L}_{1} &= 4\varpi_1 (\lambda_{\min}(P) + \lambda_{\max}(P)),\\
	\mathscr{L}_{{2_t}} &= 2\lambda_{\max}(P) \big(2\mathcal J_1^2 \varpi_1 + 2\mathcal J_1 \mathcal J_2 \varpi_2 + 2\mathcal J_1 \mathcal J_3 \varpi_3 + \mathcal J_1\sigma\\
	& ~~~~~~+ 2\mathcal J_3^2 \varpi_3 + 2\mathcal J_2 \mathcal J_3 \varpi_2 + 2\mathcal J_1 \mathcal J_3 \varpi_1 + \mathcal J_3\sigma + 2\varpi_1\tilde\mu_t\big) + 2\tilde\eta\varpi_3,
	\end{align*}
	where $\Vert A\Vert \leq \mathcal J_1$, $\Vert B\Vert \leq \mathcal J_2$,  $\Vert E\Vert \leq \mathcal J_3$, $\Vert x\Vert \leq \varpi_1$ for any $x\in X$, $\Vert \nu\Vert \leq \varpi_2$ for any $\nu\in U$, and $\Vert d\Vert \leq \varpi_3$ for any $d\in D$.
\end{lemma}	

\begin{proof}
	We first compute $\mathscr{L}_{1}$ and $\mathscr{L}_{{2_t}}$, and then take the maximum between them. By defining
	\begin{align}\notag
	\mathscr{L}_{{2_t}}\!:\left\{
	\hspace{-0.5mm}\begin{array}{l}\max\limits_{x\in X, d\in D}\Vert\frac{\partial \mathcal H_{2}}{\partial (x,d)}\Vert\\
	\,\,\,\,\,\,\, \text{s.t.} \quad \quad\!\!\Vert x\Vert \leq \varpi_1, \Vert d\Vert \leq \varpi_3,\end{array}\right.
	\end{align}
	one can reach the chain of inequalities in~\eqref{EQ:100}. For computing $\mathscr{L}_{1}$, since $\lambda_{\min}(P)\Vert x-\hat x \Vert^2 \leq (x-\hat x)^\top P(x-\hat x)$, one has $\gamma = \lambda_{\min}(P)$ in~\eqref{EQ:11}. Then we have
	\begin{align*} 
	\mathscr{L}_{1} &=\max\limits_{x\in X, \Vert x\Vert \leq \varpi_1}\Vert 2\lambda_{\min}(P) (x - \hat x) - 2P(x-\hat x)\Vert\\
	&\leq 4\varpi_1 (\lambda_{\min}(P) + \lambda_{\max}(P)).
	\end{align*}
	Then $\mathscr{L}_{\mathcal H_t} \!=\! \max\! \big\{\mathscr{L}_{1}, \mathscr{L}_{{2_t}}\big\} \!=\! \max\!\big\{4\varpi_1 (\lambda_{\min}(P) + \lambda_{\max}(P)),2\lambda_{\max}(P) (2\mathcal J_1^2 \varpi_1 + 2\mathcal J_1 \mathcal J_2 \varpi_2 + 2\mathcal J_1 \mathcal J_3 \varpi_3 + \mathcal J_1\sigma + 2\mathcal J_3^2 \varpi_3 + 2\mathcal J_2 \mathcal J_3 \varpi_2 + 2\mathcal J_1 \mathcal J_3 \varpi_1 + \mathcal J_3\sigma + 2\varpi_1\tilde\mu_t) + 2\tilde\eta\varpi_3\big\}$, which completes the proof.
\end{proof}

We now compute $\mathscr{L}_{\mathcal H_t}$ for \emph{nonlinear} control systems.

\begin{lemma}\label{Lem:1_1}
	Given a dt-NCS as in~\eqref{EQ:2}, let $(x-\hat x)^\top P(x-\hat x)$ be an APBF with a positive-definite matrix $P\in\mathbb{R}^{n\times n}$. Then $\mathscr{L}_{\mathcal H_t}$ is acquired as $\mathscr{L}_{\mathcal H_t} = \max\! \big\{\mathscr{L}_{1}, \mathscr{L}_{{2_t}}\big\}$ with 
	\begin{align*}
	\mathscr{L}_{1} &= 4\varpi_1 (\lambda_{\min}(P) + \lambda_{\max}(P)),\\
	\mathscr{L}_{{2_t}} &= 2\lambda_{\max}(P) (2\mathcal J_f \mathcal J_x + \mathcal J_x\sigma + 2\mathcal J_f \mathcal J_d + \mathcal J_d\sigma + 2\varpi_1\tilde\mu_t) + 2\tilde\eta\varpi_3,
	\end{align*}
	where $\Vert f(x,\nu,d)\Vert \leq \mathcal J_f$, $\Vert \partial_{x}f(x,\nu,d) \Vert \leq \mathcal J_{x}$, $\Vert \partial_{d}f(x,\nu,d) \Vert \leq \mathcal J_{d}$, $\Vert x\Vert \leq \varpi_1$ for any $x\in X$, and $\Vert d\Vert \leq \varpi_3$ for any $d\in D$.
\end{lemma}

\begin{proof}
	By defining
	\begin{align}\notag
	\mathscr{L}_{{2_t}}\!:\left\{
	\hspace{-0.5mm}\begin{array}{l}\max\limits_{x\in X, d\in D}\Vert\frac{\partial \mathcal H_{2}}{\partial (x,d)}\Vert\\
	\,\,\,\,\,\,\, \text{s.t.} \quad \quad\!\!\Vert x\Vert \leq \varpi_1, \Vert d\Vert \leq \varpi_3,\end{array}\right.
	\end{align}
	one can acquire the chain of inequalities in~\eqref{EQ:101}.
	For $\mathcal H_1$:
	\begin{align*} 
	\mathscr{L}_{1} &=\max\limits_{x\in X, \Vert x\Vert \leq \varpi_1}\Vert 2\lambda_{\min}(P) (x - \hat x) - 2P(x-\hat x)\Vert\\
	&\leq 4\varpi_1 (\lambda_{\min}(P) + \lambda_{\max}(P)).
	\end{align*}
	Then $\mathscr{L}_{\mathcal H_t} = \max \big\{\mathscr{L}_{1}, \mathscr{L}_{{2_t}}\big\} = \max\big\{4\varpi_1 (\lambda_{\min}(P) + \lambda_{\max}(P)),2\lambda_{\max}(P) (2\mathcal J_f \mathcal J_x + \mathcal J_x\sigma + 2\mathcal J_f \mathcal J_d + \mathcal J_d\sigma + 2\varpi_1\tilde\mu_t) + 2\tilde\eta\varpi_3\big\}$, which concludes the proof.
\end{proof}

\begin{remark}
	For the computation of $\mathscr{L}_{\mathcal H_t}$ in Lemmas~\ref{Lem:1}, \ref{Lem:1_1}, the required information is Lipschitz constant of dynamics together with an upper bound over unknown models. One can estimate the Lipschitz constant of dynamics using data based on the proposed approach in~\cite{wood1996estimation}. One can also compute an upper bound on unknown models based on the range of the state set.
\end{remark}

\begin{figure*}
	\begin{align}\notag
	&\mathscr{L}_{{2_t}} = \max\limits_{x\in X, d\in D}\Vert \begin{bmatrix}
	2(f(x,\nu,d) - \mathcal P(f(\hat x,\nu,\hat d)))^\top P\partial_{x}f(x,\nu,d)- 2\tilde\mu_t(x-\hat x)^\top P\\2(f(x,\nu,d) - \mathcal P(f(\hat x,\nu,\hat d)))^\top P\partial_{d}f(x,\nu,d)-2\eta(d-\hat d)^\top
	\end{bmatrix}\Vert\\\notag
	& \leq  \max\limits_{x\in X, d\in D} \Big\{2\Vert P \Vert\big(\Vert \partial_{x}f(x,\nu,d) \Vert(\Vert f(x,\nu,d) \Vert + \Vert \mathcal P(f(\hat x,\nu,\hat d))\Vert) + \Vert \partial_{d}f(x,\nu,d) \Vert(\Vert f(x,\nu,d) \Vert + \Vert \mathcal P(f(\hat x,\nu,\hat d))\Vert)\\\notag
	&~~~+\tilde\mu_t(\Vert x \Vert + \Vert \hat x \Vert)\big)+2\eta(\Vert d\Vert + \Vert\hat d\Vert)\Big\}\\\notag
	&\leq\max\limits_{x\in X, d\in D} \Big\{2\Vert P \Vert\big(\Vert \partial_{x}f(x,\nu,d) \Vert(\Vert f(x,\nu,d) \Vert + \sigma + \Vert f(\hat x,\nu,\hat d)\Vert) + \Vert \partial_{d}f(x,\nu,d) \Vert(\Vert f(x,\nu,d) \Vert + \sigma + \Vert f(\hat x,\nu,\hat d)\Vert)\\\notag
	&~~~+\tilde\mu_t(\Vert x \Vert + \Vert \hat x \Vert)\big)+ 2\eta(\Vert d\Vert + \Vert\hat d\Vert)\Big\}\\\label{EQ:101}
	& \leq 2\lambda_{\max}(P) (2\mathcal J_f \mathcal J_x + \mathcal J_x\sigma + 2\mathcal J_f \mathcal J_d + \mathcal J_d\sigma + 2\varpi_1\tilde\mu_t) + 2\tilde\eta\varpi_3.
	\end{align}
	\rule{\textwidth}{0.1pt}
\end{figure*}

\section{Compositional Construction of ABF for Interconnected dt-NCS}\label{Sec:Com}
Here, we provide a compositional approach to construct an ABF for an interconnected dt-NCS using its corresponding data-driven APBF of subsystems. To do so, we first raise the following $\max$-type small-gain assumption.

\begin{assumption}\label{Assump: Gamma}
	Let $\mu_{ij}\in\mathbb{R}^+$ defined as
	\begin{equation*}
	\mu_{ij} := 
	\begin{cases}
	\mu_i ~~~~& \text{if }i = j,\\
	\frac{\eta_i}{\gamma_j} ~~~~& \text{if }i \neq j,
	\end{cases}
	\end{equation*}
	satisfy
	\begin{equation}\label{Assump: small}
	\mu_{i_1i_2}.\mu_{i_2i_3}.\dots \mu_{i_{q-1}i_{q}}.\mu_{i_{q}i_1} < 1
	\end{equation}	
	for all sequences $(i_1,\dots,i_{q}) \in \{1,\dots,\mathcal M\}^ {q}$ and $q\in \{1,\dots,\mathcal M\}$. 
	
	Condition~\eqref{Assump: small} is called \emph{circularity condition} and implies the existence of $\kappa_i \in\mathbb{R}^+$  fulfilling~\cite{ruffer2010monotone}
	\begin{align}\label{compositionality}
	\max_{i,j}\Big\{\frac{\mu_{ij}\kappa_j}{\kappa_i}\Big\} < 1, ~~~~i,j = \{1,\dots,\mathcal M\}.
	\end{align}
\end{assumption}	
In the next theorem, we employ Assumption \ref{Assump: Gamma} to construct an ABF for an interconnected dt-NCS based on data-driven APBF of subsystems as in Theorem~\ref{Thm:3}.

\begin{figure*}
	\begin{align}\nonumber
	\mathcal V(\varphi,f(x,\nu),\hat{f}(\hat x, \nu))&=\max_{i}\{\frac{1}{\kappa_i}\mathcal S_{i}(\varphi_i,f_i(x_i,\nu_i,d_i),\hat{f}_{i}(\hat x_{i}, \nu_{i},\hat d_{i}))\}\\\notag
	&\leq\max_{i}\frac{1}{\kappa_i}\Big\{\max\{\mu_{i}\mathcal S_{i}(\varphi_i,x_i,\hat x_{i}),\eta_i\Vert d_i-\hat d_{i}\Vert^2,\theta_{i}\}\Big\}\\\notag
	&=\max_{i}\frac{1}{\kappa_i}\Big\{\max\{\mu_{i}\mathcal S_{i}(\varphi_i,x_i,\hat x_{i}),\eta_i\max_{j, j\neq i}\{\Vert d_{ij}-\hat d_{ij}\Vert^2\},\theta_{i}\}\Big\}\\\notag
	&=\max_{i}\frac{1}{\kappa_i}\Big\{\max\{\mu_{i}\mathcal S_{i}(\varphi_i,x_i,\hat x_{i}),\eta_i\max_{j}\{\Vert x_{j}-\hat x_{j}\Vert^2\},\theta_{i}\}\Big\}\\\notag
	&\leq \max_{i}\frac{1}{\kappa_i}\Big\{\max\{\mu_{i}\mathcal S_{i}(\varphi_i,x_i,\hat x_{i}),\eta_i\max_{j}\{\frac{\mathcal S_{j}(\varphi_j,x_j, \hat x_{j})}{\gamma_{j}}\},\theta_{i}\}\Big\}\\\notag
	&=\max_{i,j}\frac{1}{\kappa_i}\Big\{\max\{\mu_{ij}\mathcal S_{j}(\varphi_j,x_j,\hat x_{j}),\theta_{i}\}\Big\}=\max_{i,j}\frac{1}{\kappa_i}\Big\{\max\{\mu_{ij} \kappa_{j} \kappa_{j}^{-1}\mathcal S_{j}(\varphi_j,x_j,\hat x_{j}),\theta_{i}\}\Big\}\\\notag
	&\leq\max_{i,j,z}\frac{1}{\kappa_i}\Big\{\max\{\mu_{ij} \kappa_{j} \kappa_{z}^{-1}\mathcal S_{z}(\varphi_z,x_z,\hat x_{z}),\theta_{i}\Big\}=\max_{i,j}\frac{1}{\kappa_i}\Big\{\max\{\mu_{ij}\kappa_{j}\mathcal V(\varphi,x,\hat x),\theta_{i}\}\Big\}\\\label{Equ1b}
	&=\max\{\mu\mathcal V(\varphi,x,\hat x),\theta\}.\ 
	\end{align}
	\rule{\textwidth}{0.1pt}
\end{figure*}

\begin{theorem}\label{Thm: Comp}
	Consider an interconnected dt-NCS
	$\Xi=\mathcal{I}(\Xi_1,\ldots, \Xi_{\mathcal M})$ induced by $\mathcal M\in\mathbb N^+$ subsystems~$\Xi_i$. Suppose there exists an APBF between each subsystem $\Xi_i$ and its symbolic abstraction $\hat\Xi_i$ with a confidence of $1-\beta_i$, according to Theorem~\ref{Thm:3}. If Assumption~\ref{Assump: Gamma} is met,
	then
	\begin{equation}\label{Comp: ABF}
	\mathcal V(\varphi,x,\hat x) := \max_{i} \{\frac{1}{\kappa_i} \mathcal S_i(\varphi_i,x_i,\hat x_i)\}
	\end{equation}
	for $\kappa_i$ as in \eqref{compositionality}, is an ABF between $\hat \Xi=\mathcal{I}(\hat\Xi_1,\ldots, \hat\Xi_{\mathcal M})$ and $\Xi=\mathcal{I}(\Xi_1,\ldots, \Xi_{\mathcal M})$ with a confidence of $1-\sum_{i=1}^{\mathcal M}\beta_{i}$. 	
\end{theorem}

\begin{proof}
	We first show that for some $\gamma\in\R^+$, ABF $\mathcal V$ in~\eqref{Comp: ABF} satisfies condition~\eqref{EQ:5}. For all $x=[{x_1;\ldots;x_{\mathcal M}}]\in X$ and  $\hat x=[{\hat x_1;\ldots;\hat x_{\mathcal M}}]\in \hat X$, we have
	\begin{align}\notag
	\Vert x - \hat x \Vert^2&= \max_i \{\Vert x_i - \hat x_{i} \Vert^2\}\le\max_i \{\frac{\mathcal S_{i}(\varphi_i,x_i, \hat x_{i})}{\gamma_{i}}\}\\\notag
	&\le\bar\gamma\max_i \{\frac{\mathcal S_{i}(\varphi_i,x_i, \hat x_{i})}{\kappa_i}\}=
	\bar\gamma \mathcal V(\varphi,x, \hat x),
	\end{align}
	where $\bar\gamma=\max_i\big\{\frac{\kappa_{i}}{\gamma_i}\big\}$,  and condition~\eqref{EQ:5} holds with $\gamma=\frac{1}{\bar\gamma}$. 
	
	\noindent
	We now show that condition~\eqref{EQ:6} holds, as well. Let $\mu= \max_{i,j}\{\frac{\mu_{ij}\kappa_{j}}{\kappa_{i}}\}$. It follows from~\eqref{compositionality} that $\mu<1$. Then by defining $\theta$ as
	$\theta:=\max_{i}\frac{\theta_{i}}{\kappa_{i}}$, we obtain the chain of inequalities in \eqref{Equ1b}. 
	
	\noindent
	We now show that the proposed $\mathcal V$ in~\eqref{Comp: ABF} is an ABF between $\hat \Xi=\mathcal{I}(\hat\Xi_1,\ldots, \hat\Xi_{\mathcal M})$ and $\Xi=\mathcal{I}(\Xi_1,\ldots, \Xi_{\mathcal M})$ with a confidence of $1-\sum_{i=1}^{\mathcal M}\beta_{i}$. By considering events $\mathcal A_i$ as $\mathcal A_i: \Big\{\hat\Xi_i\cong_{\mathcal{S}}\!\Xi_i\Big\}$ for all $i \in\{1,\dots,\mathcal M\}$, we have $\PP\big\{\mathcal A_i\big\}\ge 1-\beta_i$ according to Theorem~\ref{Thm:3}. We aim at quantifying the concurrent occurrence of events $\mathcal A_i$:
	\begin{align}\label{eq:proof11}
	&\PP\big\{\mathcal A_1\cap\dots\cap \mathcal A_{\mathcal M}\big\}=1-\PP\big\{\bar {\mathcal A}_1\cup\dots\cup \bar {\mathcal A}_{\mathcal M} \big\},
	\end{align}
	with $\bar {\mathcal A}_i$ being complements of $\mathcal A_i, \forall i \in\{1,\dots,\mathcal M\}$. Given that
	\begin{align}\notag
	\PP&\big\{\bar {\mathcal A}_1\cup\dots\cup \bar {\mathcal A}_{\mathcal M} \big\}\leq\PP\big\{\bar {\mathcal A}_1\big\}+\dots+\PP\big\{\bar {\mathcal A}_{\mathcal M}\big\},
	\end{align}
	and by leveraging \eqref{eq:proof11}, one can finally conclude that 
	\begin{align}\label{eq:8}
	\PP\big\{\mathcal A_1\!\cap\!\dots\!\cap\! \mathcal A_{\mathcal M}\big\}\!\geq\! 1-\PP\big\{\bar {\mathcal A}_1\big\}+\dots+\PP\big\{\bar {\mathcal A}_{\mathcal M}\big\}
	\!\geq\! 1-\sum_{i =1}^{\mathcal M} \beta_i.
	\end{align}
	Hence, $\mathcal V$ is an ABF between $\hat\Xi$ and $\Xi$ with a confidence of at least $1-\sum_{i=1}^{\mathcal M}\beta_{i}$. \end{proof}

\begin{remark}
	It is worth noting that if one can synthesize $\eta_i$ and $\gamma_i$ during solving the SOP such that $\frac{\eta_i}{\gamma_j}<1$, the circularity condition~\eqref{Assump: small} is automatically fulfilled without requiring any posteriori check.  
\end{remark}

\section{Case Study: Room Temperature Network}\label{Sec:Case}
We demonstrate our data-driven results over a room temperature network composing $100$ rooms with unknown models in a circular topology, each of which is equipped with a cooler. This kind of room network is employed for storing
specific medicines in some low temperatures. The temperature evolution $x(\cdot)$ can be characterized by the following interconnected network~\cite{meyer}:
\begin{align}\notag
\Xi\!:x(k+1)=Ax(k)+\alpha T_{c}\nu(k)+ \digamma T_{E},
\end{align}
where the matrix $A$ has diagonal entries $ a_{ii}=1-2\aleph-\digamma-\alpha\nu_i(k)$, $i\in\{1,\ldots,\mathcal M\}$, off-diagonal entries $a_{i,i+1}=a_{i+1,i}=a_{1,\mathcal M}=a_{\mathcal M,1}=\aleph$, $i\in \{1,\ldots,\mathcal M-1\}$, and other entries as zero. Symbols $\aleph$, $\digamma$, and $\alpha$ are thermal factors between rooms $i \pm 1$ and $i$, the outside environment and the room $i$, and the cooler and the room $i$, respectively.
In addition, $ x(k)=[x_1(k);\ldots;x_{\mathcal M}(k)]$, $T_E=[T_{e_1};\ldots;T_{e_{\mathcal M}}]$, with  $T_{e_i}=-2\,^\circ C$, $\forall i\in\{1,\ldots,{\mathcal M}\}$, being outside temperatures. The cooler temperature is $T_c=5\,^\circ C$ and the control input is $\nu\in\{0,0.05,0.1,0.15,0.2\}$. Now by characterizing each individual room as 
\begin{align}\label{sub_room}
\Xi_i\!:x_i(k+1)&=a_{ii}{x_i}(k)+\aleph (d_{i-1}(k) + d_{i+1}(k))+ \alpha T_{c} \nu_i(k) +\digamma T_{e_i},
\end{align}
with $d_0 = d_{\mathcal M}, d_{\mathcal M+1} = d_1$, one has $\Xi=\mathcal{I}(\Xi_1,\ldots,\Xi_{\mathcal M})$. We assume the model of each room is unknown to us. The main target is to compositionally construct a symbolic abstraction as well as a data-driven ABF via solving $\text{SOP}$~\eqref{SOP}. Accordingly, we utilize the data-driven symbolic abstraction and synthesize controllers regulating the temperature of each room in a safe set $X_i = [-0.5,0.5]$ with a guaranteed probabilistic confidence. It is worth highlighting that the dimension of the sample space for each room is $n_i + p_i = 3$, since each room in the circular interconnection topology is connected to its previous and next rooms.

We consider our APBF as $\mathcal S_i(\varphi_i,x_i,\hat x_i) = \varphi_{1_i}(x_{i} - \hat x_{i})^4 + \varphi_{2_i}(x_{i} - \hat x_{i})^2 + \varphi_{3_i}$. We also fix $\varepsilon_{t_i} = 0.001$, $\beta_i = 10^{-4}$, and $\sigma_i = 0.025$, a-priori. According to~\eqref{EQ:12}, we compute the required number of data for solving $\text{SOP}$ in~\eqref{SOP} as $\mathcal Q = 776$. By solving $\text{SOP}$~\eqref{SOP} with $\mathcal Q$, we obtain the corresponding decision variables as
\begin{align*}
&\mathcal S_i(\varphi_i,x_i,\hat x_i) = 0.2(x_{i} - \hat x_{i})^4 + 0.17(x_{i} - \hat x_{i})^2 + 18,\\
&\gamma_i^* = 5.8,~ \tilde\eta^*_i = 0.02,~ \tilde\theta^*_i = 0.4,~ \xi^*_{\mathcal Q_i}=-0.3093,
\end{align*}
with a fixed $\tilde \mu_i = 0.5$. We now compute $\mathscr{L}_{\mathcal H_{t_i}} = 0.8$ according to Lemma~\ref{Lem:1_1}. We also compute $\varkappa^{-1}(\varepsilon_{t_i})$ according to Lemma~\ref{Function_g} as $\varkappa^{-1}(\varepsilon_{t_i}) = 0.3628$. Since $\xi^*_{\mathcal Q_i}+\max_t\mathscr{L}_{\mathcal H_{t_i}}\varkappa^{-1}(\varepsilon_{t_i}) = -19 \times 10^{-3} \leq 0$, the constructed data-driven $\mathcal S_i$ is an APBF between each unknown room $\Xi_i$ and its symbolic abstraction $\hat\Xi_i$ with $\gamma_i = 5.8,\mu_i = 0.995, \eta_i = 0.02, \theta_i = 0.4051,$ and a confidence of $1 - 10^{-4}$.

We now construct an ABF for the interconnected rooms using data-driven APBF of individual rooms, according to Theorem~\ref{Thm: Comp}. By taking $\kappa_i = 1, \forall i\in\{1,\dots,\mathcal M\}$, the circularity condition in~\eqref{Assump: small} is fulfilled. Hence, one can certify that $\mathcal  V(\varphi, x,\hat x) = \max_i\{\mathcal S_i(\varphi_i,x_i,\hat x_i)\} = \max_i\{0.2(x_{i} - \hat x_{i})^4 + 0.17(x_{i} - \hat x_{i})^2 + 18\}$ is an ABF between the room temperature network $\Xi$ and its symbolic abstraction $\hat\Xi$ with $\gamma = 5.8,\mu = 0.995, \theta = 0.4051,$ and a confidence of $1-\sum_{i=1}^{100}\beta_{i} = 99\%$. Accordingly based on Theorem~\ref{Thm:1}, $\mathscr{R} := \big\{(x,\hat x) \in X \times \hat X \,\big|\, \mathcal V(\varphi,x,\hat x) \leq  0.4051\big\}$
is an $\tilde\epsilon$-approximate alternating bisimulation relation between $\hat\Xi$ and $\Xi$ with $\tilde\epsilon = 0.2643$ and a confidence of $99\%$.

We now leverage the constructed data-driven symbolic abstraction and compositionally design a controller such that the controller regulates state of each unknown room in the comfort zone $[-0.5,0.5]$. To do so, we first synthesize a controller for each abstract room $\hat \Xi_i$ via \texttt{SCOTS} \cite{rungger2016scots} and then refine it back over unknown original room $\Xi_i$. The overall controller for the network is then a vector whose entries are controllers for individual rooms. Closed-loop state trajectories and their corresponding control inputs of a representative room are depicted, respectively, in Figs.~\ref{Simulation} and~\ref{Simulation1}. As observed, the designed controller maintains trajectories of an unknown representative room within the safe set $[-0.5,0.5]$. It is noteworthy that we have considered the basis functions $g_{j_i}(x_i,\hat x_i)$ as monomials over $x_i$ and $\hat x_i$. Consequently, the APBF is treated as a polynomial, given that models of unknown room temperatures are inherently polynomial in nature, in accordance with their underlying physics. It is important to emphasize that our approach is applicable to general class of nonlinear systems, capable of enforcing general temporal logic properties using the proposed data-drive abstractions. The room temperature example here is provided solely for the purpose of illustrating the results.
\begin{figure}[h]
	\centering 
	\includegraphics[width=0.55\linewidth]{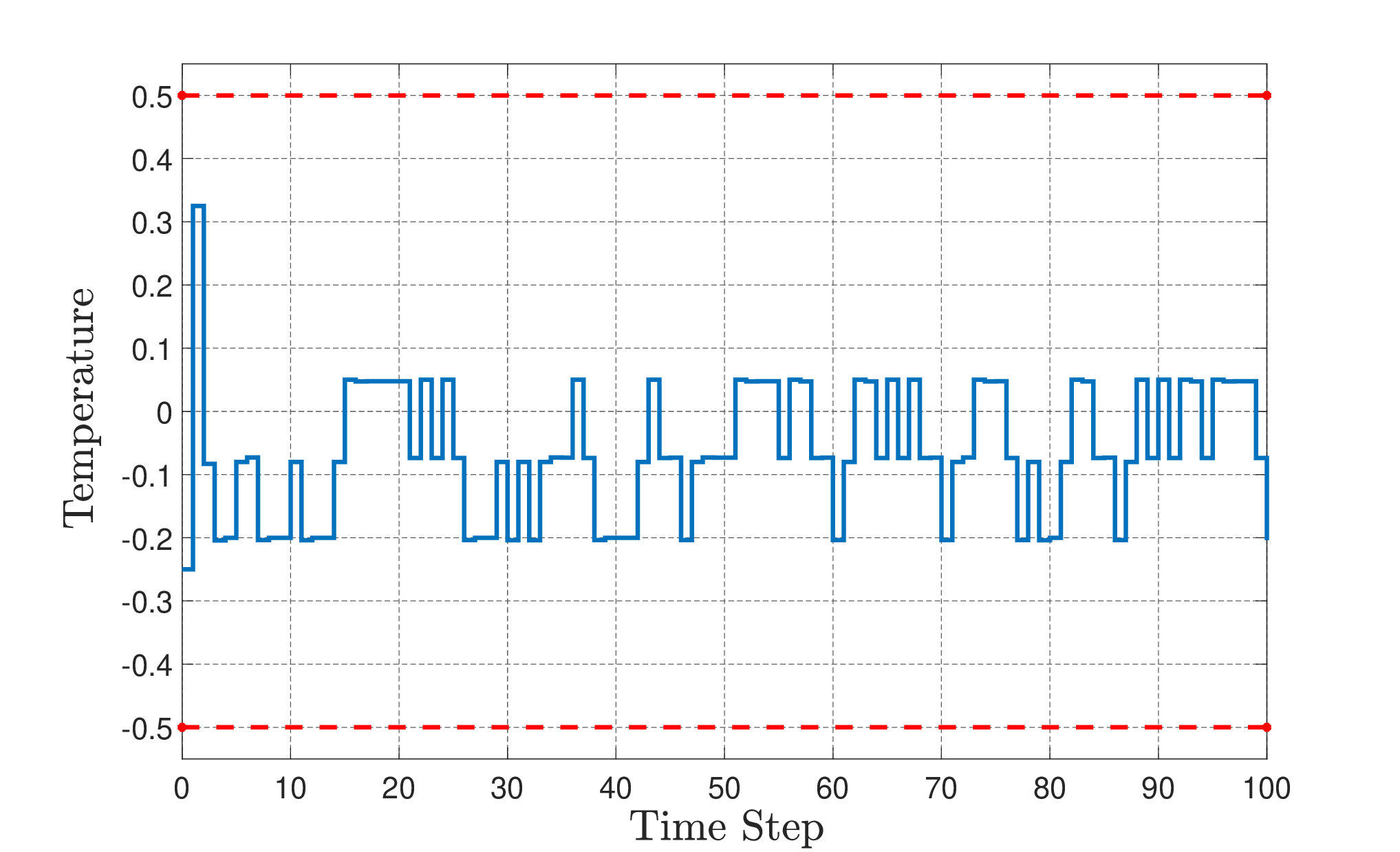}
	\caption{Closed-loop state trajectories of a representative room by designing the controller over its data-driven symbolic abstraction.}
	\label{Simulation}
\end{figure}

\begin{figure}[h]
	\centering 
	\includegraphics[width=0.55\linewidth]{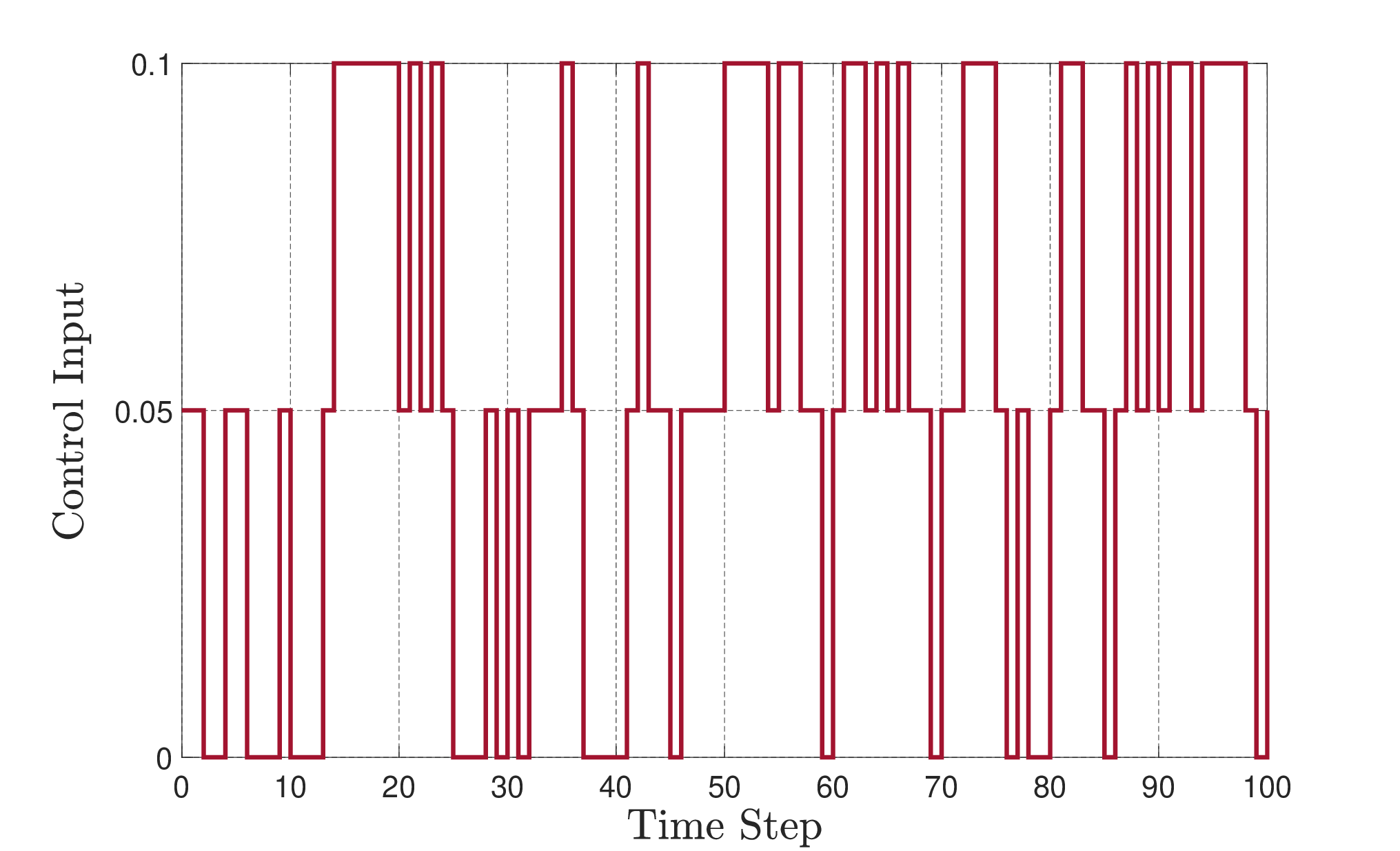}
	\caption{A synthesized control input for a representative room via its data-driven symbolic abstraction.}
	\label{Simulation1}
\end{figure}

\section{Conclusion}
In this article, we developed a data-driven divide-and-conquer approach using small-gain reasoning to construct symbolic abstractions for interconnected control networks with unknown mathematical models. We first built a relation between each unknown subsystem and its data-driven symbolic abstraction using alternating pseudo-bisimulation functions (APBF), while providing a guaranteed probabilistic confidence. We then proposed a compositional approach via $\max$-type small-gain reasoning to construct an \emph{alternating bisimulation function} for an unknown interconnected network using its data-driven APBF of subsystems. We illustrated the efficacy of our data-driven results over a room temperature network composing $100$ rooms with unknown models.

\bibliographystyle{alpha}
\bibliography{biblio}

\end{document}